\documentclass[oribibl]{llncs}
\usepackage{amsmath}
\usepackage{mathdots}
\usepackage{mathrsfs}
\usepackage{amssymb}
\usepackage{bm}

\usepackage{multirow}
\usepackage{makecell,multirow,diagbox}
\usepackage{xcolor}

\usepackage[ruled]{algorithm2e}

\title{Improved Key Generation Algorithm for Gentry's Fully Homomorphic Encryption Scheme}

\author{Yang Zhang\and Renzhang Liu\and Dongdai Lin}

\institute{State Key Laboratory of Information Security,\\
    Institute of Information Engineering\\
    Chinese Academy of Sciences, Beijing 100093, China.\\
    \email{liurenzhang@iie.ac.cn}
    }

\date{}

\begin{document}
\maketitle

\begin{abstract}
At EUROCRYPT 2011, Gentry and Halevi implemented a variant of Gentry's fully homomorphic encryption scheme. The core part in their key generation is to generate an odd-determinant ideal lattice having a particular type of Hermite Normal Form. However, they did not give a rigorous proof for the correctness. We present a better key generation algorithm, improving their algorithm from two aspects.
\begin{itemize}
\item We show how to deterministically generate ideal lattices with odd determinant, thus increasing the success probability close to 1.
\item We give a rigorous proof for the correctness. To be more specific, we present a simpler condition for checking whether the ideal lattice has the desired Hermite Normal Form. Furthermore, our condition can be checked more efficiently.
\end{itemize}
As a result, our key generation is about 1.5 times faster. We also give experimental results supporting our claims. Our optimizations are based on the properties of ideal lattices, which might be of independent interests.

\end{abstract}

\section{Introduction}

    Fully homomorphic encryptions (FHE) support arbitrary operations on encrypted data which have a wide range of applications such as private information retrieval\cite{retrieval1}, electronic watermark\cite{watermark}. Ever since the introduction of the concept of FHE by Rivest et al \cite{RAD78}, there have appeared many homomorphic encryption schemes which support limited operations\cite{AD97,RSA78}. Nevertheless, all these schemes failed to be FHE.

    It was not until 2009 that Gentry \cite{Gen09} proposed the first plausible FHE scheme using ideal lattice. Since then, researchers proposed many FHE schemes \cite{BGV12,MW15,GSW13}. Nowadays, all existing FHE schemes follow Gentry's blueprint. Specifically, one first constructs a ``somewhat homomorphic" scheme that supports some limited operations on the encrypted data. Then he ``squashes" the decryption procedure so that it can be expressed as operations supported by the scheme. Finally, he converts the ``somewhat homomorphic" scheme into a fully homomorphic scheme. The low efficiency of FHE is one of the main obstacles that prevents it from being practical.

    The first attempt to implement Gentry's FHE scheme was made by Smart and Vercauteren\cite{SV10} in 2010. They gave an optimized version of Gentry's scheme, which decreased the key size by a linear factor. However, they needed to generate ideal lattices with prime determinant during the key generation. In practice, one may need to try as many as $n^{1.5}$ candidates before finding one with prime determinant when working with lattices in dimension $n$. As a result, they failed to implement Gentry's fully homomorphic encryption scheme.

    In 2011, Gentry and Halevi\cite{GH11} first completely implemented Gentry's fully homomorphic encryption scheme continuing in the same direction of the implementation of \cite{SV10}. The strong requirement that the lattice has a prime determinant was removed. Instead, they only required that the determinant is odd and the lattice has a simple-HNF (see Definition 2). They found in practice that the success probability was roughly 0.5. Thus one needs to try two candidates on average to get a valid key. Moreover, they proposed a method to check whether a lattice has a simple-HNF. However, they did not provide a rigorous proof for the correctness of their algorithm and some details on their implementation were not very clear.

    Gentry-Halevi's implementation depends heavily on the underlying algebraic structure of the scheme, that is, the 2-power cyclotomic fields. Scholl and Smart\cite{SS11} showed how to generalize the fast key generation techniques to arbitrary cyclotomic fields. They also obtained a key generation algorithm which is roughly twice faster.

    \subsubsection{Our Contribution} By studying the properties of ideal lattices, we present further improvements on the key generation algorithm in \cite{GH11}. Our algorithm has a rigorous proof for correctness and is about 1.5 times faster in practice. Our focus is on Gentry-Halevi's original implementation, which is different from \cite{SS11}.

    As stated before, the core idea in the key generation algorithm of \cite{GH11} is to generate an odd-determinant ideal lattice, which has a simple-HNF.

    In order to generate ideal lattice with odd determinant, the authors of \cite{GH11} chose a random generator and computed the determinant of the generated ideal lattice. Since the ideal lattice under consideration is highly structured and we only care about the parity of the determinant, it might be possible to determine whether the determinant is odd without computing it. In fact, we do find that the parity of the determinant is connected to the generator in a very simple way. Thus we can generate odd-determinant ideal lattice deterministically by imposing a simple constraint on the generator.

    To check whether the ideal lattice has the desired HNF, they gave a simple checking condition, which avoided computing the HNF explicitly out of the consideration of efficiency. By studying and exploiting the properties of the HNF of ideal lattice, we are able to present another checking condition. By verifying that our condition holds, we can rigorously prove that our algorithm outputs an ideal lattice with the desired HNF. Furthermore, our condition can be checked more efficiently.

    We need to point out that our improvements are of more theoretical than practical significance due to the following reasons. First, Gentry and Halevi's implementation is the very first one, and there are much better implementations of different schemes \cite{Sho,NTRU,DM15,GHS15}. Besides, key generation is not a bottleneck of FHE schemes and the algorithm in \cite{GH11} is very fast even for very large parameters, which makes our slight speedup less significant.

    Nevertheless, it is important to make things more clear. Besides, our improvements are based on some new, special properties of the HNF of ideal lattices, which we believe are of independent interests.

    \subsubsection{Roadmap} The paper is organized as follows. In Section 2, some preliminaries are presented. In Section 3, we revisit the key generation algorithms in previous implementations. In Section 4, we propose our optimized key generation for Gentry's FHE scheme, together with some theoretical analyses and experimental results. A brief conclusion will be given in the final section.

\section{Preliminaries}
\subsubsection{Notations}
We use bold capital and lowercase letters to denote matrices and vectors respectively. For a matrix $\bm{M}\in \mathbb{R}^{n\times m}$, the $i$-th row is denoted as $\bm{M}_{i}$ and the entry in $(i,j)$-th entry is denoted as $m_{i,j}.$ For a vector $\bm{x}\in \mathbb{R}^{n}$, the $i$-th entry is denoted as ${x}_{i}$. These notations are used throughout the paper unless specified otherwise.

\subsection{Hermite Normal Form}
For integer matrices, there is a very important standard form known as the Hermite Normal Form.
\begin{definition}[Hermite Normal Form]\label{HNF}
A nonsingular matrix $\bm{H} \in \mathbb{Z}^{n\times n}$ is said to be in HNF, if
    \begin{itemize}
    \item $h_{i,i} > 0$ for $1 \leq i \leq n$.
    \item $h_{j,i} = 0$ for $1\leq j < i \leq n$.
    \item $0 \leq h_{j,i} < h_{i,i} $ for $1 \leq i < j\leq n$.
    \end{itemize}
\end{definition}
For ``randomly" generated matrices of high dimensions, the diagonals of its HNF are highly unbalanced: most of them are small (in fact most of them are 1), and the first diagonal is usually the only large one. We can define a very special HNF.
\begin{definition}[Simple-HNF]
A nonsingular matrix $\bm{H}\in \mathbb{Z}^{n\times n}$ is said to be in simple-HNF if it is in HNF and $h_{i,i}=1$ for $2\leq i\leq n$.
\end{definition}
It has been shown that asymptotically the density of ``simple-HNF" for randomly generated $n\times n$ matrices (in the sense that each entry is chosen uniformly at random from $\{-M,-M+1,\cdots, M\}$ for large enough $M$) is $\frac{1}{\prod_{j=2}^{n}\zeta(j)}\approx 44\%$\cite{Hu2016,Maz2011}, where $\zeta(j)$ is the Riemann zeta function.

\subsection{Lattice}
A lattice is a discrete subgroup of $\mathbb{R}^m$. Formally,
\begin{definition}[Lattice]
Given $n$ linearly independent vectors $\bm{B}=\{\bm{b}_{1},\bm{b}_{2},\cdots,\bm{b}_{n}\}$, where $\bm{b}_{i}\in \mathbb{R}^m$, the lattice $\mathcal{L}(\bm{B})$ generated by $\bm{B}$ is defined as following $$\mathcal{L}(\bm{B})=\{\sum_{i=1}^{n} x_{i}\bm{b}_{i}:x_{i}\in \mathbb{Z}\}=\{\bm{x}\bm{B}:\bm{x}\in\mathbb{Z}^{n}\}.$$
\end{definition}
We call $\bm{B}$ a basis of $\mathcal{L}(\bm{B})$, $m$ and $n$ the dimension and the rank of $\mathcal{L}(\bm{B})$ respectively. When $m=n$, we say $\mathcal{L}(\bm{B})$ is full-rank.

\begin{definition}[Determinant]
For lattice $\mathcal{L}(\bm{B})$, the determinant is defined as $$\det(\mathcal{L})=\sqrt{\det(\bm{BB}^T)}.$$
\end{definition}
When $\mathcal{L}(\bm{B})$ is full-rank, $\bm{B}$ is nonsingular and $\det(\mathcal{L})=|\det(\bm{B})|$.

\begin{definition}[Primitive Lattice Vector]
A lattice vector $\bm{v}\in\mathcal{L}$ is called a primitive lattice vector if for any integer $k>1$, $\bm{v}/k \notin \mathcal{L}$.
\end{definition}
There is an easy criterion for determining whether a lattice vector is primitive.
\begin{proposition}
Given a lattice $\mathcal{L}$ with basis $\bm{B}$, $\bm{v}=\bm{xB} \in \mathcal{L}$ is primitive if and only if $\gcd(x_1,\cdots,x_n)=1$.
\end{proposition}

\subsection{Ideal Lattice}
In what follows, we focus on $ R=\mathbb{Z}[x]/\left \langle f(x)\right \rangle$, where $f(x)\in \mathbb{Z}[x]$ is a monic polynomial of degree $n$ and  $\left \langle f(x)\right \rangle$ is the ideal generated by $f(x)$ in $\mathbb{Z}[x]$.

Consider the coefficient embedding
\begin{align*}
  \sigma:  R  &\rightarrow \mathbb{Z}^n,  \\
        \sum_{i=0}^{n-1} a_{i}x^{i} &\mapsto (a_{0},a_{1},\cdots,a_{n-1}).
\end{align*}

For any polynomial $v(x)\in R$, we use $\bm{v}$ to denote the image of $v(x)$ under $\sigma$, and consider them equivalent.

Let $g(x)$ be a polynomial of degree $< n$ and $v(x)\in \left \langle g(x)\right \rangle$, then there exists a polynomial $w(x)\in \mathbb{Z}[x]$ of degree $< n$ such that $v(x)=w(x)g(x)\mod f(x)$. Write $w(x)=w_{n-1}x^{n-1}+\cdots+w_{1}x+w_{0}$, then $v(x)=\sum_{i=0}^{n-1} w_{i}x^ig(x)\mod f(x)$. So each element in $\left \langle g(x)\right \rangle$ is a linear combination of $g(x)\mod f(x)$, $xg(x)\mod f(x)$, $\cdots$, $x^{n-1}g(x)\mod f(x)$. It follows that $\left \langle g(x)\right \rangle$ is a lattice under the coefficient embedding. Thus, we can define
\begin{definition}[Ideal Lattice]
For $ g(x) \in R=\mathbb{Z}[x]/\left \langle f(x)\right \rangle$, where $f(x)\in \mathbb{Z}[x]$ is a monic polynomial of degree $n$, the ideal generated by $g(x)$ forms a lattice $\mathcal{L}$ under the coefficient embedding. We call $\mathcal{L}$ the ideal lattice generated by $g(x)$.
\end{definition}

Moreover, if $g(x)$ and $f(x)$ are coprime over $\mathbb{Q}$ (hence over $\mathbb{Z}$ since $f(x)$ is monic), $g(x)\mod f(x),xg(x)\mod f(x),\cdots,x^{n-1}g(x)\mod f(x)$ are linearly independent. Otherwise there exist integers $y_{0},y_{1},\cdots,y_{n-1}$  not all zero, such that $\sum_{i=0}^{n-1} y_{i}x^ig(x)\mod f(x)=0$, that is, $(\sum_{i=0}^{n-1}y_{i}x^i)g(x)=0\mod f(x)$, which indicates that $g(x)$ and $f(x)$ are not coprime. Therefore, the ideal lattice generated by $g(x)$ is full-rank.

\subsection{Resultant}
The resultant of two polynomials is defined as
    \begin{definition}[Resultant]
     Let $a(x)=a_{m}x^m+\cdots+a_{1}x+a_{0}$, $b(x)=b_{n}x^n+\cdots+b_{1}x+b_{0}$ $\in \mathbb{R}[x]$. Define the Sylvester matrix of $a(x)$ and $b(x)$ as
    $$\mbox{Sylv}(a(x),b(x))=\begin{bmatrix}
                                                                           a(x) \\
                                                                           \vdots \\
                                                                           x^{n-1}a(x) \\
                                                                           b(x) \\
                                                                           \vdots \\
                                                                           x^{m-1}b(x)
                                                                         \end{bmatrix} =\begin{bmatrix}
       a_{0} & a_1 & \cdots & a_{m}& &  &  &  \\
        & a_{0} & a_1 & \cdots & a_{m}&  &  &  \\
        &   & \ddots  &  & &\ddots  & & \\
        &  &  & a_{0} & a_1 &\cdots & & a_{m}\\
       b_{0} & &\cdots & &b_{n} & & & \\
        &  \ddots& && & &\ddots & \\
        &  &  b_{0}& &  \cdots & &  &b_{n}
     \end{bmatrix}_{(n+m)\times(n+m)}.$$
     Then the resultant of $a(x)$ and $b(x)$, denoted as $Res(a,b)$, is the determinant of $\mbox{Sylv}(a(x),b(x))$.
    \end{definition}

\section{Key Generation Algorithms of Gentry's FHE Scheme}
In this section we revisit the key generation algorithms implementing Gentry's fully homomorphic encryption scheme in \cite{SV10,GH11}.

\subsection{The Smart-Vercauteren's Key Generation}
    Smart and Vercauteren \cite{SV10} first attempted to implement Gentry's FHE schemes. For the key generation, they fixed a monic irreducible polynomial $f(x)\in \mathbb{Z}[x]$ of degree $n$. Then they repeatedly chose another polynomial $g(x)$ randomly from some set until the resultant of $f(x)$ and $g(x)$ was prime. They showed that the HNF of such ideal lattice had the following simple-HNF.
    \begin{center}
      $\begin{bmatrix}
        d & 0 & &0 &   &  & 0 \\
        [-r]_{d}  & 1 & &0 &   & &  0 \\
        [-r^2]_{d}  & 0 & &1 &   &  & 0  \\
          &   &   & & \ddots&  &  \\
        [-r^{n-1}]_{d} & 0 & &0  &  &  & 1
      \end{bmatrix}.$
    \end{center}
    where $[x]_{d}=\ x\mod d$. Such lattices can be represented implicitly by two integers $d$ and $r$, thus the key size is reduced by a factor of $n$.

    However, they failed to implement the FHE scheme since they were unable to generate keys for lattice of dimension $n>2048$. One of the main obstacles lies in the inefficiency of generating ideal lattices with prime determinant by trial and error.

\subsection{The Gentry-Halevi's Key Generation}
    In \cite{GH11}, Gentry and Halevi found that it was not necessary for the determinant to be prime. Instead, they showed that the scheme can go through as long as the ideal lattice has an odd determinant and a simple-HNF. The key generation consists of the following steps.

    \begin{enumerate}
      \item[(i)] Fix $f(x)=x^n+1$ with $n$ a power of $2$. Then they choose a vector $\bm{v}=(v_{0},\cdots,v_{n-1})$, where each entry $v_i$ is chosen at random as a $t$-bit integer. Consider the ideal lattice generated by $v(x)=v_{n-1}x^{n-1}+\cdots+v_{1}x+v_{0}$ in $\mathbb{Z}[x]/\left \langle f(x)\right \rangle$. Then the lattice is generated by the following basis
          \begin{center}
            $\bm{V}=\begin{bmatrix}
            v_{0} & v_{1} & v_{2} &   & v_{n-1} \\
            -v_{n-1} & v_{0} & v_{1} &   & v_{n-2} \\
            -v_{n-2} & -v_{n-1} & v_{0} &   & v_{n-3} \\
            &   &   & \ddots&   \\
            -v_{1}& -v_{2}& -v_{3} &   & v_{0}
            \end{bmatrix}.$
          \end{center}

      \item[(ii)] Given $v(x)$ and $f(x)$, use their highly optimized algorithm to compute the resultant $d$ of $v(x)$ and $f(x)$ and the constant term $w_{0}$ of $w(x)$, where $w(x)=w_{n-1}x^{n-1}+\cdots+w_{1}x+w_{0}$ is the unique polynomial such that $v(x)w(x)=d\mod f(x)$.  If $d$ is odd, compute $w_1$ by applying their highly optimized algorithm to $(xv(x)\mod f(x))$ and $f(x)$.
      \item[(iii)] If $\gcd(w_{1},d)=1$, compute $r=\frac{w_{0}}{w_{1}}\mod d$, check whether $r^n=-1\mod d$. If so, they consider that the ideal lattice has the desired HNF, then they compute an odd coefficient $w_i$ via $w_{i}=rw_{i+1} \mod d$ and $w_0$, $w_1$ (The subscripts are modulo $n$.) and use $(d,r)$ as the public key and $w_{i}$ the secret key.
    \end{enumerate}

For the sake of clarity and comparison, the key generation algorithm is summarized in Algorithm \ref{gentry keygen}. We need to emphasize that there are some details missing in their description (Lines 3-5 and Lines 7-9 in Algorithm \ref{gentry keygen}) and Algorithm \ref{gentry keygen} is adopted from their source code.
    \begin{algorithm}\label{gentry keygen}
    \caption{Key Generation in \cite{GH11}}
    \LinesNumbered
    \KwIn{dimension $n$, bit length $t$.}
    \KwOut{pk=$(d,r)$, sk=$w_i$.}
    Choose a random $n$-dimensional vector $\vec{v}$, where each $v_{i}$ is a $t$-bit integer.\\
    Compute the resultant $d$ of $v(x)$ and $f(x)$, and the constant term $w_{0}$ of $w(x)$, where $w(x)v(x)=d\mod f(x)$.\\
    \If{$d$ is even}{Go to 1.}
    Compute the coefficient $w_{1}$ of $w(x)$.\\
    \If{$w_{1}$ has no inverse modulo $d$}{Go to 1.}
    Compute $r=\frac{w_{0}}{w_{1}}\mod d$.\\
    \eIf{$r^n\neq -1\mod d$}{Go to 1.\\}
    {Compute an odd $w_i$ via $w_{i}=rw_{i+1} \mod d$ and $w_0$, $w_1$. (The subscripts are modulo $n$.)\\
    Output: pk=$(d,r)$, sk=$w_i$.\\}
    \end{algorithm}

\subsection{Remarks on Gentry-Halevi's Key Generation}\label{remarks}
    The authors of \cite{GH11} solved the main issues that prevented the key generation in \cite{SV10} from being practical. They also gave many optimizations focusing on practical performance. However, we note that there are several natural questions regarding to their implementation.
    \begin{enumerate}
      \item In \cite{GH11}, the authors mentioned that ``it was observed by Nigel Smart that the HNF has the correct form whenever the determinant is odd and square-free. Indeed, in our tests this condition was met with probability roughly 0.5..." Since the determinants in the experiments are numbers with millions of bits, it is difficult to determine whether they are square-free or not. We believe that ``this condition" means ``the HNF has the correct form" rather than ``the determinant is odd and square-free". We also rerun the experiment to confirm our assertion.
      \item In step (iii), they did not mention what to do if $\gcd(w_{1},d) \neq 1$. Generally, there are two options. We can start over by choosing another lattice, or we can compute another random coefficient of $w(x)$ until we find one that is coprime to $d$. (Judging from their code, they did the former.) In fact, we will show if $\gcd(w_{1},d) \neq 1$, then all the coefficients of $w(x)$ are not coprime to $d$. That means the latter is not optional.
      \item Also in step (iii), when $\gcd(w_{1},d) = 1$, we compute $r$, and check if $r^n=-1\mod d$ holds. Even if this equality held, they did not show that the HNF has the desired form. In other words, they did not provide a proof for the correctness of their algorithm. They did mention to check a serial of conditions that could guarantee the HNF is simple, however, in their implementation they actually tested only the last condition. It is vital for the FHE scheme to work correctly. We will give a simpler condition and by checking our condition holds, we can guarantee the HNF is indeed simple.
    \end{enumerate}


%

\section{Improved Key Generation for Gentry's FHE Scheme}
In this section, we present our improved key generation algorithm for Gentry's FHE scheme. Our improvements consist of two aspects:
\begin{itemize}
\item we can generate ideal lattices with odd determinant deterministically.
\item we present a rigorous proof for the correctness of our algorithm by checking a simpler sufficient and necessary condition for the ideal lattice having a simple-HNF.
\end{itemize}

Before presenting our key generation algorithm, we first give some useful lemmas and propositions on ideal lattices.

\subsection{Some Properties of Ideal Lattices}
The following lemma was already stated without proof in \cite{GH11}.
\begin{lemma}
Let $\mathcal{L}$ be the ideal lattice generated by $g(x) \in R=\mathbb{Z}[x]/\left \langle f(x)\right \rangle$, where $f(x)$ is a monic polynomial of degree $n$ and $g(x)$ is of degree $m$. Assume $g(x)$ is relatively prime to $f(x)$, then $\det(\mathcal{L})=|Res(g,f)|=|Res(f,g)|.$
\end{lemma}

\begin{proof}
Note that
$$Sylv(g(x),f(x))=\begin{bmatrix}
                   g(x) \\
                   \vdots \\
                   x^{n-1}g(x) \\
                   f(x) \\
                   \vdots \\
                   x^{m-1}f(x)
                   \end{bmatrix} = \begin{bmatrix}
       g_{0} & g_1 & \cdots & g_{m}& &  &  &  \\
        & g_{0} & g_1 & \cdots & g_{m}&  &  &  \\
        &   & \ddots  &  & &\ddots  & & \\
        &  &  & g_{0} & g_1 &\cdots & & g_{m}\\
       f_{0} & &\cdots & &f_{n} & & & \\
        &  \ddots& && & &\ddots & \\
        &  &  f_{0}& &  \cdots & &  &f_{n}
     \end{bmatrix}.$$
Since $f_n=1$, the Sylvester matrix can always be transformed unimodularly into the following (block-triangular) form by adding proper multiples of rows in the lower half to each row on the top half,
$$\begin{bmatrix}
        &  & \bm{B}_{n\times n} &  &  & \bm{0} & & \\\hline
       f_{0} & &\cdots & &f_{n} & & & \\
        &  \ddots& && & &\ddots & \\
        &  &  f_{0}& &  \cdots & &  &f_{n}
     \end{bmatrix}=\begin{bmatrix}
           g(x) \mod f(x) \\
           \vdots \\
           x^{n-1}g(x) \mod f(x)\\\hline
           f(x) \\
           \vdots \\
           x^{m-1}f(x)
           \end{bmatrix}.$$
    It follows that $Res(g,f) = f_n^m\det(\bm{B})=\det(\bm{B})$. \\
    Since $g(x)$ is relatively prime to $f(x)$, $\mathcal{L}$ is a full-rank lattice with basis $\bm{B}$. Therefore, $\det(\mathcal{L})= |\det(\bm{B})| = |Res(g,f)| = |Res(f,g)|.$\qed
    \end{proof}

The next lemma shows that the HNF of an ideal lattice has a very special structure.
    \begin{lemma}\label{divide}
    Let $\mathcal{L}$ be the ideal lattice generated by $g(x) \in R=\mathbb{Z}[x]/\left \langle f(x)\right \rangle$, where $f(x)$ is a monic polynomial of degree $n$ and is relatively prime to $g(x)$. Then the Hermite Normal Form of $\mathcal{L}$
      \begin{center}
      $\bm{H}=\begin{bmatrix}
      h_{1,1} & &  &  &  \\
      h_{2,1} & h_{2,2} & & & \\
      \vdots &  \vdots & \ddots& &  \\
      h_{n,1} & h_{n,2}&  \cdots& & h_{n,n}
    \end{bmatrix}$
    \end{center}
      satisfies $h_{i,i}|h_{j,l}$, for $1\leq l\leq j\leq i\leq n$.
    \end{lemma}
    \begin{proof}
     We prove the lemma by induction on $i$.\\
     For $i=1$, it is trivial.\\
     For $i=2$, consider the first row, which corresponds to the constant polynomial $h_{1,1}$. Since $\mathcal{L}$ is an ideal lattice, the vector $(0,h_{1,1},0\cdots,0)$, which corresponds to the polynomial $h_{1,1}x$, is also in $\mathcal{L}$.\\ It's obvious that $(0,h_{1,1},0\cdots,0)=x_{1}\bm{H}_{1}+x_{2}\bm{H}_{2}$ for some $x_{1},x_{2} \in \mathbb{Z}.$
     Then $h_{1,1}=x_{2}h_{2,2}$, $x_{1}h_{1,1}+x_{2}h_{2,1}=0$. Hence $h_{2,2}|h_{1,1},h_{2,2}|h_{2,1}$, which completes the proof for $i=2$.

    Assume the result holds for $i\leq k \leq n-1$, $h_{i,i}|h_{j,l}, $ where $1\leq l \leq j\leq i\leq k.$  We show that for $i=k+1$, $h_{k+1,k+1}|h_{j,l}$.

    Consider the $k$-th row. The corresponding polynomial of $k$-th row is $$h_{k,k}x^{k-1}+h_{k,k-1}x^{k-2}+\cdots+h_{k,2}x+h_{k,1}.$$ After multiplying $x$, we get a vector $ (0,h_{k,1},\cdots,h_{k,k},0,\cdots,0),$ which is a linear combination of $\bm{H}_{1},\cdots,\bm{H}_{k+1}$ with integer coefficients, i.e.
    \begin{center}
      $(0,h_{k,1},\cdots,h_{k,k},0,\cdots,0)$=$\sum_{i=1}^{k+1} y_{i}\bm{H}_{i}$,
    \end{center}
    where $y_{i} \in \mathbb{Z}$, for $i=1,\cdots,k+1.$

    So
    \begin{equation}\label{set1}
      \begin{cases}
        h_{k,k}&=y_{k+1}h_{k+1,k+1}\\
            h_{k,k-1}&=y_{k}h_{k,k}+y_{k+1}h_{k+1,k}\\
             &\vdots\\
            h_{k,1}&=\sum_{i=2}^{k+1} y_{i}h_{i,2}\\
            0&=\sum_{i=1}^{k+1} y_{i}h_{i,1}
      \end{cases}.
    \end{equation}

 From the first equation, we get $y_{k+1}=\frac{h_{k,k}}{h_{k+1,k+1}}$ and
  \begin{equation*}
    \begin{cases}
         h_{k+1,k}&=\frac{h_{k,k-1}-y_{k}h_{k,k}}{h_{k,k}}h_{k+1,k+1}\\
         h_{k+1,k-1}&=\frac{h_{k,k-2}-y_{k-1}h_{k-1,k-1}-y_{k}h_{k,k-1}}{h_{k,k}}h_{k+1,k+1}\\
            &\vdots\\
         h_{k+1,2}&=\frac{h_{k,1}-\sum_{i=2}^{k}{y_{i}h_{i,2}}}{h_{k,k}}h_{k+1,k+1}\\
         h_{k+1,1}&=\frac{-\sum_{i=1}^{k}{y_{i}h_{i,1}}}{h_{k,k}}h_{k+1,k+1}
         \end{cases}.
  \end{equation*}

  From the induction hypothesis, we have $h_{k,k}|h_{j,l}$ for $1\leq l\leq j\leq k \leq n.$ So the coefficient of $h_{k+1,k+1}$ in each equation is in fact an integer, therefore $h_{k+1,k+1}|h_{k+1,l}$, $1\leq l\leq k+1.$
  Since $h_{k+1,k+1}|h_{k,k}$, we know $h_{k+1,k+1}|h_{j,l}$, where $1\leq l \leq j \leq k+1 \leq n.$ Thus, the result holds for $i=k+1$.

  By the principle of induction, the lemma follows.\qed
\end{proof}

 \begin{remark}
   A part of the result has already been proven in \cite{Ding} for identifying ideal lattices. In fact, they proved that the diagonal entries of an HNF form a division chain, that is,  $h_{n,n}|\cdots|h_{2,2}|h_{1,1}.$ Our results show that off-diagonal entries are also divisible by the corresponding diagonal entry in the same row, which means the primitive part of the polynomial\footnote{The primitive part of an integer polynomial $s(x)$ is $s(x)/r$, where $r$ is the g.c.d of the coefficients of $s(x)$.} corresponding to each row is monic.
 \end{remark}

  \begin{remark}\label{HNFcharacterization}
   For the HNF of a full-rank lattice, the first row is \textit{the primitive lattice vector} of the form $(d,0,\cdots,0)$, where $d>0$ is a factor of the determinant. In the case of ideal lattice, the first row corresponds to the smallest positive constant polynomial in the ideal. This characterization will be useful for proving Proposition \ref{checking condition}.
 \end{remark}

\begin{lemma}\label{root}
  Let the notations be the same as in Lemma \ref{divide} and let $H_{i}(x)$ denote the corresponding polynomial of $i$-th row in the HNF, $i\leq n$. Let $\beta$ be the root of $h_{2,1}+h_{2,2}x$. Then $H_{i}(\beta)$ =0 $\mod \frac{h_{1,1}h_{i,i}}{h_{2,2}}$, $\forall i \geq 2$ and $f(\beta)=0\mod \frac{h_{1,1}}{h_{2,2}}$.
\end{lemma}
\begin{proof}
   We prove the first equality by induction on $i$.

    For $i=2$, by the definition of $\beta,$ $H_{2}(\beta)=0\mod h_{1,1}.$

    Assume $H_{i}(\beta)=0\mod \frac{h_{1,1}h_{i,i}}{h_{2,2}},$ for $i\leq k,$ where $2\leq k \leq n-1.$  From Equation \ref{set1}, we have

    \begin{equation*}
    \begin{cases}
         h_{k,k}x^{k}&=y_{k+1}h_{k+1,k+1}x^{k}\\
         h_{k,k-1}x^{k-1}&=y_{k}h_{k,k}x^{k-1}+y_{k+1}h_{k+1,k}x^{k-1}\\
             &\vdots\\
         h_{k,1}x&=\sum_{i=2}^{k+1} y_{i}h_{i,2}x\\
         0&=\sum_{i=1}^{k+1} y_{i}h_{i,1}
         \end{cases}.
  \end{equation*}

    Sum the equations up,
    $$y_{k+1}H_{k+1}(x)+y_{k}H_{k}(x)+y_{k-1}H_{k-1}(x)+\cdots+y_{1}H_{1}(x)=xH_{k}(x),$$
    Set $x=\beta$,
    $$y_{k+1}H_{k+1}(\beta)+(y_{k}-\beta)H_{k}(\beta)+y_{k-1}H_{k-1}(\beta)+\cdots+y_{1}H_{1}(\beta)=0.$$

    Note that $H_{1}(x)=h_{1,1},\ H_{1}(\beta)=0\mod \frac{h_{1,1}h_{k,k}}{h_{2,2}}. $ By induction hypothesis, $H_{i}(\beta)=0\mod \frac{h_{1,1}h_{i,i}}{h_{2,2}}$ and $\frac{h_{1,1}h_{i+1,i+1}}{h_{2,2}}|\frac{h_{1,1}h_{i,i}}{h_{2,2}}$,  $y_{k+1}H_{k+1}(\beta)=0\mod \frac{h_{1,1}h_{k,k}}{h_{2,2}}.$
    Since $y_{k+1}=\frac{h_{k,k}}{h_{k+1,k+1}},$ we have $$H_{k+1}(\beta)=0\mod \frac{h_{1,1}h_{k+1,k+1}}{h_{2,2}}.$$

    Therefore for $i=k+1$, the equality also holds. Thus $H_{i}(\beta)$ =0 $\mod \frac{h_{1,1}h_{i,i}}{h_{2,2}}$, $\forall i \geq 2$.

    For the second equality, note that $xH_{n}(x)\mod f(x)=xH_{n}(x)- h_{n,n}f(x)$. Since the vector corresponding to $xH_{n}(x)-h_{n,n} f(x)$ is a lattice vector, there exist integers $z_1\cdots z_n \in \mathbb{Z}$, $xH_{n}(x)-h_{n,n} f(x)=\sum_{i=1}^{n}z_{i}H_{i}(x)$.

    Set $x=\beta$, $$\beta H_{n}(\beta)-h_{n,n} f(\beta)=\sum_{i=1}^{n}z_{i}H_{i}(\beta).$$ Since $h_{n,n}|h_{i,i}$ for all $i$, $H_{i}(\beta)$ =0 $\mod \frac{h_{1,1}h_{n,n}}{h_{2,2}}$, $\forall i \geq 2$. Also $H_{1}(\beta)$ =0 $\mod \frac{h_{1,1}h_{n,n}}{h_{2,2}}$. Then
    $h_{n,n} f(\beta)=0\mod \frac{h_{1,1}h_{n,n}}{h_{2,2}}$ and  $$f(\beta)=0\mod \frac{h_{1,1}}{h_{2,2}}.$$\qed
\end{proof}

\subsection{Generating Ideal Lattice with Odd Determinant Deterministically}
The key generation algorithm in \cite{GH11} needs to generate an ideal lattice with odd determinant. The original paper achieved this by trial and error. We first show that the parity of the determinant is related to the generator in a very simple way, and then give a deterministic way to generate an ideal lattice with odd determinant.

In the next two subsections, we fix $f(x)=x^n+1$ with $n$ a power of 2.

    \begin{proposition}\label{parity}
       Let $\mathcal{L}$ be the ideal lattice generated by $v(x)=v_{n-1}x^{n-1}+\cdots+v_{1}x+v_{0}$ in $\mathbb{Z}[x]/\left \langle f(x)\right \rangle$. Then $\det(\mathcal{L})\equiv v_{0}+v_{1}+\cdots +v_{n-1} \mod 2.$
    \end{proposition}
    \begin{proof}
    Since we only concern the parity of the determinant, we work over $\mathbb{F}_{2}$.
    Denote
    \[\bm{P}=\left[ {\begin{array}{*{20}{c}}
        0&1&{}&0&0\\
        0&0&{}&0&0\\
        {}&{}& \ddots &{}&{}\\
        0&0&{}&0&1\\
        1&0&{}&0&0
    \end{array}} \right].\]
    Then
    \begin{center}
      $\bm{V}=\begin{bmatrix}
        v_{0} & v_{1} & v_{2} &   & v_{n-1} \\
        -v_{n-1} & v_{0} & v_{1} &   & v_{n-2} \\
        -v_{n-2} & -v_{n-1} & v_{0} &   & v_{n-3} \\
          &   &   & \ddots&   \\
        -v_{1}& -v_{2}& -v_{3} &   & v_{0}
      \end{bmatrix} = {v_{0}}\bm{I} +  \cdots + {v_{n-1}}{\bm{P}^{n - 1}} = v(\bm{P}) \mbox{ over }\mathbb{F}_{2}.$
    \end{center}

    Now we compute the eigenvalues of $\bm{P}$.

    Since $\bm{P}$ is a cyclic shift matrix, $\bm{P}^n=\bm{I}$. Note that $x^n+1=(x+1)^n$ over $\mathbb{F}_{2}$ ($n$ is a power of $2$), then all the eigenvalues of $\bm{P}$ are 1. All the eigenvalues of $\bm{V}$ are thus $v(1)=v_{0}+v_{1}+\cdots+v_{n-1}$. Hence,
    $$\det(\mathcal{L})\equiv (v_{0}+v_{1}+\cdots+v_{n-1})^n \equiv v_{0}+v_{1}+\cdots+v_{n-1}\mod 2. $$\qed
    \end{proof}

    From the above proposition, we know how to generate ideal lattices with odd determinant deterministically. Specifically, we choose a random $n$-dimensional vector $\bm{v}$ from the set $$\{v(x) \in \mathbb{Z}[x]/\left \langle f(x)\right \rangle: v_i \text{\ is a random\ } t \text{-bit integer and \ } \sum_{i=0}^{n-1}v_i\equiv 1\mod 2 \}.$$  This can be done by choosing $v_i$ randomly, and use $v(x)+(v(1)+1\mod 2)$ as the generator or simply use $2u(x)+1$, where $u(x)$ is a random vector. Then the generated ideal lattice has an odd determinant.

    \subsection{A Simpler Condition for Checking the HNF}
     In this part, we give a simpler condition, which yields a rigorous proof for the correctness of our algorithm. Our condition is based on the following proposition.
    \begin{proposition}\label{checking condition}
      Let $\mathcal{L}$ be the ideal lattice generated by $v(x)$ in $\mathbb{Z}[x]/\left \langle f(x)\right \rangle$. Suppose $w(x)\in \mathbb{Z}[x]/\left \langle f(x)\right \rangle$ is a polynomial such that $v(x)w(x)=d \mod f(x)$ where $d$ is the determinant of $\mathcal{L}$. Then the following conditions are equivalent.
      \begin{enumerate}
      \item[(1)] $\mathcal{L}$ has a simple-HNF.
      \item[(2)] $\mathcal{L}$ contains a vector of the form $\vec{r}=(-r,1,0,\cdots,0)$.
      \item[(3)] There exists an index $i$, $ 0\leq i \leq n-1$, $\gcd(w_{i},d)=1$.
      \item[(4)] For arbitrary $0\leq i \leq n-1$, $\gcd(w_{i},d)=1$.
      \end{enumerate}
    \end{proposition}

      \begin{proof}
    \hspace{0.1cm}
    \begin{enumerate}
      \item[(1)$\Leftrightarrow$(2)] The equivalence between the two conditions was proved in \cite{GH11}.\\

      \item[(2)$\Rightarrow$(3)] Assuming first that $\vec{r}=(-r,1,0,\cdots,0) \in \mathcal{L}$. Then there exists $y(x) \in \mathbb{Z}[x]/\left \langle f(x)\right \rangle$, such that $y(x)v(x)=r(x)\mod f(x)$. Therefore, $$y(x)v(x)w(x)=r(x)w(x)\mod f(x).$$ $$dy(x)=r(x)w(x)\mod f(x).$$
          Note that $f(x)=x^n+1$, we have $$(-w_{n-1},\cdots,w_{n-3},w_{n-2})-r(w_{0},\cdots,w_{n-2},w_{n-1})=0\ \mod\ d.$$
          So $w_{0}r=-w_{n-1} \mod d$ and $w_{i+1}r=w_{i} \mod d$, for $0\leq i \leq n-2$.\\
          We prove by contradiction that $\exists\ 0\leq i \leq n-1$, $\gcd(w_{i},d)=1$.\\ If for arbitrary $0\leq i \leq n-1$, $\gcd(w_{i},d)\neq 1$, let $\mu=\gcd(w_{0}, d) > 1$. From the relations among the $w_i$'s, we know $\mu$ divides all the $w_i$'s. Therefore, $\mu | \gcd(w_0, \cdots, w_{n-1}, d)$. Hence $\frac{d}{\mu}=\frac{w(x)}{\mu} v(x) \mod f(x)$ is a lattice vector, which means the first diagonal of the HNF is a proper factor of $d$. So the second diagonal can't be 1, otherwise the determinant is $\frac{d}{\mu}$ rather than $d$. This is a contradiction. Therefore $\exists\ 0\leq i \leq n-1$, $\gcd(w_{i},d)=1$.\\

      \item[(3)$\Rightarrow$(4)] Assume $\exists\ 0\leq i \leq n-1$, $\gcd(w_{i},d)=1$, fix $i$. We also prove by contradiction. Suppose there exists a $0\leq j \leq n-1$ such that $\mu=\gcd(w_{j}, d)> 1$.

           Due to Lemma \ref{divide}, we can assume the second row of HNF is $(-\alpha r,\alpha,0,\cdots,0)$, for some $\alpha \in \mathbb{N^{+}}$. Then $\alpha^2|\alpha h_{1,1}|d$ and $\gcd(\alpha, w_i)=1$. Similar to previous proof,
           $$\alpha(-w_{n-1},\cdots,w_{n-3},w_{n-2})-\alpha r(w_{0},\cdots,w_{n-2},w_{n-1})=0\ \mod\ d.$$
           Hence $$(-w_{n-1},\cdots,w_{n-3},w_{n-2})-r(w_{0},\cdots,w_{n-2},w_{n-1})=0\ \mod\ \frac{d}{\alpha}.$$
           According to the steps above, $w_{i+1}r=w_{i} \mod \frac{d}{\alpha}$, for $i \leq n-2$ and $w_{0}r=-w_{n-1} \mod \frac{d}{\alpha}$. Since $\alpha^2|d$, $\frac{d}{\alpha}$ and $d$ share exactly the same prime factors, hence $\gcd(w_{j},\frac{d}{\alpha})=\mu'>1$ . Similar to the previous proof, we have $\mu'$ divides all the coefficients of $w(x)$. Specifically, $\mu'|w_{i}$. Hence $\mu'|\gcd(w_{i},d)$, a contradiction. Thus $\gcd(w_{i},d)=1$ for any $0\leq i \leq n-1$.\\

      \item[(4)$\Rightarrow$(2)] Assume for any $0\leq i \leq n-1,\ \gcd(w_{i},d)=1$. Since $d=w(x)v(x)\mod f(x)$, $\gcd(w_0,\cdots,w_{n-1})|d$. Hence  $\gcd(w_0,\cdots,w_{n-1})=1$. Then $(d,0,\cdots,0)$ is a primitive lattice vector in $\mathcal{L}$. According to Remark \ref{HNFcharacterization}, it's the first row of the HNF of $\mathcal{L}$, which means all the other diagonals in the HNF are 1. Thus, $\mathcal{L}$ contains a vector (the second row of its HNF) of the form $\vec{r}=(-r,1,0,\cdots,0)$.
    \end{enumerate}
    \qed
    \end{proof}

    From the above proposition, we can conclude that in order to check whether the ideal lattice has a simple-HNF, we only need to compute a $w_i$ and $\gcd(w_i, d)$, which is simpler and more efficient to check. If $\gcd(w_i, d)=1$, then from the proposition, the lattice has a simple-HNF. Otherwise, the lattice does not have a simple-HNF.

    In what follows, we show that once $w_1$ has an inverse modulo $d$, it must hold that $r^n=-1\mod d$. Thus the checking can be safely left out.



\begin{proposition}
    Let $\mathcal{L}$ be the ideal lattice generated by $v(x)$ in $\mathbb{Z}[x]/\left \langle f(x)\right \rangle$. Suppose $w(x)\in \mathbb{Z}[x]/\left \langle f(x)\right \rangle$ is a polynomial such that $v(x)w(x)=d \mod f(x)$ where $d$ is the determinant of $\mathcal{L}$. Assume that $\gcd(w_1,d)=1$ and $r=\frac{w_0}{w_1}\mod d$, then $r^n=-1\mod d$.
\end{proposition}
\begin{proof}
From the assumptions and the proof for Proposition \ref{checking condition}, we know that $(-r,1,0,\cdots,0)$ is in $\mathcal{L}$. According to Lemma \ref{root}, $r$ is a root of $f(x)=0\mod d$. Therefore $r^n = -1\mod d$.\qed
\end{proof}

\subsection{Our Key Generation Algorithm}

    We formally present our improved key generation algorithm (Algorithm \ref{our keygene}) and give some theoretical analyses and experimental results supporting our claims in this section.

    \begin{algorithm}\label{our keygene}
    \caption{Our Key Generation Algorithm}
    \LinesNumbered
    \KwIn{dimension $n$, bit length $t$.}
    \KwOut{pk=$(d,r)$, sk=$w_i$.}
    Choose a random vector $\vec{v}$ from $\{v(x)=\sum_{i=0}^{n-1}v_ix^i: v_i \text{\ is a random\ } t \text{-bit integer and \ } \sum_{i=0}^{n-1}v_i\equiv 1\mod 2 \}.$\\
    Compute the resultant $d$ of $v(x)$ and $f(x)$, and coefficient $w_{1}$ of the linear term of $w(x)$, where $w(x)v(x)=d\mod f(x)$.\\
    \eIf{$\gcd(w_{1},d)\neq 1$}{Go to 1.\\}{Compute $w_{0}$ and $r=\frac{w_{0}}{w_{1}}\mod d$.\\
    Compute an odd $w_i$ via $w_{i}=rw_{i+1} \mod d$ and $w_0$, $w_1$. (The subscripts are modulo $n$.)\\
    Output: pk=$(d,r)$, sk=$w_{i}$.\\}
    \end{algorithm}

\subsubsection{Theoretical Analyses}

We first give time estimations for both algorithms. Here we show a few notations to represent the time for different routines.
\begin{itemize}
\item $T_{res}$ : time to compute the resultant and $w_i$ using the algorithm in \cite{GH11}. Note that in both algorithms, we sometimes only need to compute the resultant or a single coefficient $w_i$, this will also take time $T_{res}$ since the exactly same routine is used to achieve these goals.
\item $T_{xgcd}$ : time to apply the extended Euclidean algorithm.
\item $T_{pmod}$ : time to compute $r^n\mod d$.
\item $T_{mul}$ : time to do one multiplication modulo $d$.
\item $T_{oddcoe}$ : time to compute an odd coefficient of $w(x)$ (line 14 in Algorithm \ref{gentry keygen} and line 7 in our algorithm).
\end{itemize}

For Algorithm \ref{gentry keygen}, since the probability of choosing an odd-determinant ideal lattice is 0.5, we need to try twice on average to get an odd-determinant ideal lattice, which costs time $2T_{res}$ in order to compute the determinants. At the same time, we have also computed a coefficient $w_{0}$. Afterwards, we need to compute another coefficient $w_{1}$, which costs time $T_{res}$. From the experiment, we found that the probability of an odd-determinant ideal lattice having simple-HNF is very high, we simply omit the failure cases. Then computing the inverse of $w_{1}$ and $r$ takes time $T_{xgcd}$ and $T_{mul}$ respectively, and checking if $r^n=-1\mod d$ takes time $T_{pmod}$. Finally we need to choose an odd coefficient of $w(x)$, which costs time $T_{oddcoe}$. So the expected running time of Algorithm \ref{gentry keygen} is $3T_{res}+T_{xgcd}+T_{pmod}+T_{mul}+T_{oddcoe}$.

In our algorithm, the ideal lattice always has an odd determinant. As before, we omit the failure cases where an odd-determinant ideal lattice doesn't have a simple-HNF, then we only need to compute two consecutive coefficients of $w(x)$, which costs $2T_{res}$. Then it costs time $T_{xgcd}$+$T_{mul}$ to perform an extended Euclidean algorithm and multiplication modulo $d$ to compute $r$. At last the time for choosing an odd coefficient of $w(x)$ is $T_{oddcoe}$. Thus the total time used in our key generation algorithm is $2T_{res}+T_{xgcd}$+$T_{mul}+T_{oddcoe}$.

In our experiments, we found that $T_{res}$ takes most proportion of the whole time. So on average, our algorithm is about 1.5 times faster.

\subsubsection{Experimental Results}
We present our experimental results in the following. Our experiments were performed on a PC (Intel(R) Core(TM) i7, 3.4GHz, 2G RAM) and based on Shoup's NTL library\cite{Sho} (version 9.10.0). NTL is installed with MPC (version 1.0.2), GMP (version 5.1.3), and MPFR (version 3.1.5).
\paragraph{1.} We rerun the experiments of \cite{GH11} to confirm our assertion in Section \ref{remarks}. We generated 100 random ideal lattices for each parameter set, and count the number of ideal lattices in different categories. In the table, the first row are the parameters $(n, t)$, like $(512, 380)$, where $n$ is the dimension of the ideal lattice and $t$ is the bit length of each coefficient in the generator. ``even (or odd) $d$" indicates the generated lattice has even (or odd) determinant, and ``(non-)SHNF" indicates the lattice has (doesn't have) a simple-HNF.

\begin{center}
\begin{tabular}{|c|c|c|c|c|c|c|c|c|}
 \hline
 \multirow{2}{*}{} & \multicolumn{2}{c|}{(512,380)} & \multicolumn{2}{c|}{(2048, 380)} & \multicolumn{2}{c|}{(8192,380)} & \multicolumn{2}{c|}{(32768, 380)}\\
 \cline{2-9} & Alg. \ref{gentry keygen} & Alg. \ref{our keygene} & Alg. \ref{gentry keygen} & Alg. \ref{our keygene} & Alg. \ref{gentry keygen} & Alg. \ref{our keygene} & Alg. \ref{gentry keygen} & Alg. \ref{our keygene}\\
 \hline
 {even $d$, SHNF} & 25  & 0 & 32 & 0 & 25 & 0 &24 & 0\\
 \hline
 {even $d$, non-SHNF}& 27  & 0 & 25 & 0 & 18 &0 & 25&0 \\
\hline
 {odd $d$, SHNF}& \textcolor{red}{48} & \textcolor{red}{98} & \textcolor{red}{42} & \textcolor{red}{98} & \textcolor{red}{57} & \textcolor{red}{100} & \textcolor{red}{47} & \textcolor{red}{90}\\
 \hline
 {odd $d$, non-SHNF}& 0 & 2 & 1 & 2 & 0 &0 & 4& 10\\
 \hline
\end{tabular}
\end{center}

From the experiment we can see that the probability of a randomly generated ideal lattice in Algorithm \ref{gentry keygen} having the correct form (odd-determinant and simple-HNF) is roughly 0.5 and that the failure cases is usually due to the determinant being even.

Also, there is something interesting here. When the ideal lattices have an even determinant, about half of them have a simple-HNF. It is not surprising if the probability of a random ideal lattice having simple-HNF is somewhat higher than $44\%$ due the divisibility of the diagonal entries of its HNF. However, when the ideal lattices have an odd determinant, most of them also have a simple-HNF. As a result, the probability of simple-HNF is about $75\%$, which is much higher than $44\%$. We conjecture this might has something to do with the underlying 2-power cyclotomic fields.
\paragraph{2.} We present some experimental results on the running time of both algorithms. In the experiments, we generated 20 valid keys and count the averaged time for the different parts in theoretical analysis. In the next table, the ``\#(trials)" column denotes the number of trials we did to get $20$ valid keys. $t_{res}$, $t_{xgcd}$, $t_{pmod}$, $t_{mul}$ and $t_{oddcoe}$ are averaged time for doing the corresponding operations, so $t_{res}\approx 3T_{res}$ in Algorithm \ref{gentry keygen}. $t_{total}$ is the averaged time to generate a valid key. All the timings are measured in seconds.
\begin{center}
  \begin{tabular}{|c|c|c|c|c|c|c|c|c|}
    \hline
    \multirow{2}*{dim $n$} & \multirow{2}*{bit-size $t$} & \multicolumn{7}{c|}{Gentry and Halevi's Algorithm} \\
    \cline{3-9}
     &  & \#(trials) & $t_{res}$ & $t_{xgcd}$ &$t_{pmod}$& $t_{mul}$ & $t_{oddcoe}$ & $t_{total}$  \\
     \hline
    512 & 380 & 37 & 0.171 & 0.048 & 0.054 & 0.007 & 0.003 & 0.284 \\
    \hline
    2048 & 380 & 46 & 1.585 & 0.354 & 0.347 & 0.037 & 0.009 & 2.336 \\
    \hline
    8192 & 380 & 36 & 12.068 & 2.120 & 2.040 & 0.173 & 0.086 & 16.500 \\
    \hline
    32768 & 380 & 46 & 123.152 & 14.128 & 14.494 & 0.962 & 0.799 & 153.591 \\
    \hline
    \multirow{2}*{dim $n$} & \multirow{2}*{bit-size $t$} & \multicolumn{7}{c|}{Our Algorithm}  \\
    \cline{3-9}
      &   & \#(trials) & $t_{res}$ & $t_{xgcd}$  & $t_{mul}$ & $t_{oddcoe}$ & $t_{total}$ & speedup \\
    \hline
    512 & 380 & 20 & 0.118 & 0.047 & 0.007 & 0.002 & 0.174 & \color{red}{1.632x}  \\
    \hline
    2048 & 380 & 20 & 0.966 & 0.351 & 0.038 & 0.036 & 1.393 & \color{red}{1.677x}  \\
    \hline
    8192 & 380 & 20 & 7.568 & 2.108 & 0.173 & 0.069 & 9.924 & \color{red}{1.663x} \\
    \hline
    32768 & 380 & 22 & 78.276 & 13.882 & 0.964 & 0.432 & 93.592 & \color{red}{1.641x} \\
    \hline
  \end{tabular}
\end{center}

We can see that the ratio between $t_{res}$ in Algorithm \ref{gentry keygen} and $t_{res}$ in Algorithm \ref{our keygene} is about 1.5, which matches the theoretical prediction ($3T_{res}/2T_{res}$). In two algorithms, $t_{xgcd}$ and $t_{mul}$ are almost the same and $t_{oddcoe}$ is very short. We note that this is because in most cases one of $w_0$ and $w_1$ is odd. From the experiments, we can see that our key generation is about 1.5 (more close to 1.6) times faster.

\section{Conclusion}

    In this paper, we made improvements on the key generation of \cite{GH11}. Using the properties of ideal lattice, we present an improved key generation algorithm with a rigorous proof for correctness. As a result, our algorithm is about 1.5 times faster.


\end{document}